\documentclass[12pt, a4paper, reqno]{amsart}
\usepackage{mathrsfs}
\usepackage{bbm}
\usepackage{amsmath,amscd,amssymb,latexsym}
\usepackage[all]{xy}
\usepackage{color}
\usepackage[section]{placeins}
\usepackage[latin1]{inputenc}
\input xypic

\evensidemargin=\oddsidemargin
\DeclareMathVersion{can}
\DeclareMathAlphabet{\can}{OT1}{cmss}{m}{n}
\newtheorem{thm}{Theorem}[section]
\newtheorem{cor}[thm]{Corollary}
\newtheorem{lem}[thm]{Lemma}

\newtheorem{rem}[thm]{Remark}
\newtheorem{exa}[thm]{Example}
\theoremstyle{definition}
\newtheorem{defn}[thm]{Definition}
\theoremstyle{fact}

\theoremstyle{conjecture}

\numberwithin{equation}{section}

\newtheorem{definition}{Definition}

\begin{document}

\title[ ]{  MDS or NMDS self-dual codes from twisted generalized Reed-Solomon codes}

\author[Huang] {Daitao Huang}
\address{\rm
College of Computer Science and Technology, Nanjing University of Aeronautics and Astronautics,
Nanjing,  211100, P. R. China}
\email{dthuang666@163.com}

\author[Yue]{Qin Yue}
\address{\rm Department of Mathematics, Nanjing University of Aeronautics and Astronautics,
Nanjing,  211100, P. R. China}
\email{yueqin@nuaa.edu.cn}

\author[Niu] {Yongfeng Niu}
\address{\rm
College of Computer Science and Technology, Nanjing University of Aeronautics and Astronautics,
Nanjing,  211100, P. R. China}
\email{niuajm@163.com}

\author[Li]{Xia Li}
\address{\rm Department of Mathematics, Nanjing University of Aeronautics and Astronautics,
Nanjing,  211100, P. R. China}
\email{lixia4675601@163.com}

\thanks{This work was supported in part by National Natural Science Foundation of China (No. 61772015).  }




\begin{abstract}
Self-dual maximum distance separable codes  (self-dual MDS codes) and self-dual near MDS codes are very important in coding theory and practice.  Thus, it is interesting to construct self-dual MDS or self-dual near MDS codes. In this paper, we not only give check matrices of dual codes of  twisted generalized  Reed-Solomon codes (TGRS codes) but also present the efficient and necessary condition of self-dual TGRS codes.
Moreover, we construct several classes of self-dual MDS or self-dual near MDS codes from TGRS codes.

\end{abstract}
\maketitle

\section{Introduction}
 In recent years, study of self-dual maximum distance separable( MDS for short)  codes have attracted a lot of attention \cite{5}-\cite{14}.
First of all, MDS codes achieve optimal parameters that allow correction of maximal number of errors for a given code rate. Study of various properties of MDS codes, such as classification \cite{15,20} of MDS codes, non-Reed-Solomon MDS codes \cite{21} and existence of MDS codes \cite{77}, has been the center of the area. In addition, MDS codes are closely connected to combinatorial design and finite geometry \cite{18,8}. Furthermore, the generalized Reed-Solomon codes are a class of MDS codes and have found wide applications in practice. On the other hand, self-dual codes, due to their nice structures,  have been attracting attention from both coding theorists, cryptographers and mathematicians. Or rather, self-dual codes have various applications in cryptography \cite{55}-\cite{19} and combinatorics \cite{18,8}.
Thus, it is natural to consider the intersection of these two classes, namely, MDS self-dual codes. For example, some new self-dual MDS codes from generalized Reed-Solomon codes are constructed in \cite{3}. And some other new self-dual MDS codes are also constructed in \cite{27}-\cite{ZZZ}.

Similar to self-dual MDS codes, self-dual near MDS (NMDS for short) codes have nice structures as well. NMDS codes were introduced in 1995 in \cite{555} by weakening the definition of MDS codes. If a code has one singleton defect from being an MDS code, then it is called almost MDS (AMDS). An AMDS code is an NMDS code if the dual code is also an AMDS code.
NMDS codes also have application in secret sharing scheme \cite{24}. In \cite{23}, NMDS codes are constructed by properties of elliptic curves.
MDS self-dual codes over large prime fields that arise from the solutions of systems of diophantine equations are constructed  and  many self-dual MDS (or near-MDS) codes of lengths up to 16 over various prime fields are constructed in \cite{26}.
From both theoretical and practical points of view, it is natural to study self-dual NMDS codes. Although there has been lots of work on NMDS codes in literature, little is known for self-dual NMDS codes. It seems challenging to construct self-dual NMDS codes.

Different from Reed-Solomon codes, twisted Reed-Solomon codes are firstly introduced in \cite{1} and a efficient and necessary condition for twisted Reed-Solomon codes to be MDS was given. In this paper, we mainly constructed self-dual  MDS codes and self-dual  NMDS codes by twisted generalized Reed-Solomon codes. Moreover, a efficient and necessary condition for twisted generalized Reed-Solomon codes to be  NMDS is naturally given in this paper and
we constructed several classes of self-dual  MDS and self-dual NMDS codes from
twisted generalized Reed-Solomon codes.

The paper is organized as follows. In section 2, there are some basic notations and results about twisted generalized Reed-Solomon codes.  Several classes of self-dual MDS codes are constructed in Section 3. In Section 4, we conclude the paper.

\section{Preliminaries}
In this section, we review some basic notations and some basic knowledge.  In particular,  we introduce MDS codes and NMDS codes from twisted generalized Reed-Solomon codes and show their check matrices.

\subsection{TGRS codes}
Let $\mathbb{F}_q[x]$  be the polynomial ring over a field field $\Bbb F_q$ of order $q$.
We denote the rank of a  matric $M$ over $\mathbb{F}_q$ by $R(M)$.
We abbreviate generalized Reed-Solomon codes,  twisted Reed-Solomon codes,  and  twisted generalized  Reed-Solomon codes as RS codes, TRS codes, and   TGRS codes, respectively.

Now, let us recall some definitions of  TRS codes (see \cite{1}).

\begin{definition}
Let $\mathcal V$ be a $k$-dimensional $\mathbb{F}_q$-linear subspace of  $ \mathbb{F}_q[x]$. Let $\alpha_1,\ldots,\alpha_n $ be distinct elements in $\mathbb{F}_q$ and  $\alpha=(\alpha_1,\ldots,\alpha_n)$. Let $v_1,\ldots, v_n$ be nonzero elements in $\Bbb F_q$ and $v=(v_1,\ldots, v_n)$. We call $\alpha_1,\ldots,\alpha_n$ the evaluation points. Define the evaluation map of $\alpha$ on $\mathcal{V}$ by
$$ev_{\alpha}:\mathcal{V}\longrightarrow \mathbb{F}_q^n, f(x)\longmapsto ev_\alpha(f(x))=(f(\alpha_1),\ldots, f(\alpha_n));$$
define the evaluation map of $\alpha$ and $v$ by
$$ev_{\alpha, v}: \mathcal V\longrightarrow\mathbb F_{q}^n, f(x)\longmapsto ev_{\alpha,v}(f(x))=(v_1f(\alpha_1),\ldots, v_nf(\alpha_n)).$$
\end{definition}
\begin{definition}
Let $k$, $t$, and $h$ be positive integers with  $0\leq h < k \leq q$ and  $\eta \in \Bbb F_q^*=\mathbb{F}_q\backslash\{0\}$. Define the set of $(k,t,h,\eta)$-twisted polynomials as
$$\mathcal{V}_{k,t,h,\eta}=\{f(x)=\sum_{i=0}^{k-1}a_ix^i+\eta a_hx^{k-1+t}:a_i\in \mathbb{F}_q, 0\le i\le k-1\},$$
which is a $k$-dimensional $\Bbb F_q$-linear subspace.  We call $h$ the hook and $t$ the twist.
\end{definition}

In this paper, we always assume that $h=k-1$ and  $t=1$, so $$\mathcal{V}_{k,1,k-1,\eta}=\{f(x)=\sum_{i=0}^{k-1}a_ix^i+\eta a_{k-1}x^{k}:a_i\in \mathbb{F}_q, 0\le i\le k-1\}.$$ For convenience, set $k\le n-k$.

\begin{defn}
Let $\alpha_1,\ldots,\alpha_n $ be distinct elements in $\Bbb F_q$ and  $\alpha=(\alpha_1,\ldots,\alpha_n)$.  Let $v_1,\ldots, v_n$ be nonzero elements in $\Bbb F_q$ and $v=(v_1,\ldots, v_n)$. Let $\mathcal V_{k,1,k-1,\eta}$ be in   Definition $2$. The TRS code of length $n$ and dimension $k$ is defined as
$$\mathcal{C}_k(\alpha,1,\eta)=ev_{\alpha}(\mathcal{V}_{k,1,k-1,\eta})\subseteq \Bbb F_q^n.$$
The TGRS code of length $n$ and dimension $k$ is defined as
$$\mathcal C_{k}(\alpha, v, \eta)=ev_{\alpha,v}(\mathcal V_{k,1,k-1,\eta})\subseteq \Bbb F_q^n.$$
\end{defn}
In fact, if $v=(1,\ldots, 1)=1$, then $\mathcal C_{k}(\alpha, v, \eta)=\mathcal{C}_k(\alpha,1,\eta)$, i.e., the TGRS code is the TRS code.

 Let  $G_k$ be a  generator matrix  of  $\mathcal C_{k}(\alpha, v, \eta)$, then

\begin{eqnarray}\label{2.1}
  G_k &=&\left(
  \begin{array}{cccc}
    v_1 & v_2 & \ldots & v_n \\
    v_1\alpha_1 & v_2\alpha_2 & \ldots & v_n\alpha_n \\
     \vdots & \vdots &  & \vdots \\
    v_1(\alpha_1^{k-1}+\eta \alpha_1^k) & v_2(\alpha_2^{k-1}+\eta \alpha_2^k) & \ldots & v_n(\alpha_n^{k-1}+\eta \alpha_n^k) \\
  \end{array}
\right).
\end{eqnarray}

 We shall find the check matrix of   of $\mathcal C_{k}(\alpha, v, \eta)$.

\begin{thm} \label{111} Let  $G_k$ be  a  generator matrix  of  $\mathcal C_{k}(\alpha, v, \eta)$ in (\ref{2.1}).    Write $u_i=\prod_{j=1,j\neq i}^n(\alpha_i-\alpha_j)^{-1}$, $1\le i\le n$ and
$a=\sum_{i=1}^n\alpha_i$.

(1) Suppose that  $a\ne 0$ and $\eta \neq -a^{-1}$.  Then

\begin{eqnarray}\label{2.2}
  H_{n-k} &=& \tiny\left(
  \begin{array}{cccc}
   \frac{ u_1}{v_1}  & \ldots & \frac{u_n}{v_n} \\
    \frac{u_1}{v_1}\alpha_1  & \ldots & \frac{u_n}{v_n}\alpha_n \\
    \vdots  &  & \vdots \\
    \frac{u_1}{v_1}\alpha_1^{n-k-2}  & \ldots & \frac{u_u}{v_n}\alpha_1^{n-k-2} \\
   \frac{ u_1}{v_1}(\alpha_1^{n-k-1}- \frac{\eta }{1+a\eta}\alpha_1^{n-k}) & \ldots &\frac{ u_n}{v_n}(\alpha_n^{n-k-1}-\frac{\eta }{1+a\eta}\alpha_n^{n-k}) \\
  \end{array}
\right)
\end{eqnarray}
is the check matrix of $C_k(\alpha, v, \eta)$,

  (2) Suppose that  $a=0$ and $\eta \neq 0$.  Then

\begin{eqnarray}\label{2.22}
  H_{n-k} &=& \tiny\left(
  \begin{array}{cccc}
   \frac{ u_1}{v_1} & \frac{u_2}{v_2} & \ldots & \frac{u_n}{v_n} \\
    \frac{u_1}{v_1}\alpha_1 & \frac{u_2}{v_2}\alpha_2 & \ldots & \frac{u_n}{v_n}\alpha_n \\
    \vdots & \vdots &  & \vdots \\
    \frac{u_1}{v_1}\alpha_1^{n-k-2} & \frac{u_2}{v_2}\alpha_2^{n-k-2} & \ldots & \frac{u_u}{v_n}\alpha_1^{n-k-2} \\
   \frac{ u_1}{v_1}(\alpha_1^{n-k-1}- \eta \alpha_1^{n-k}) & \frac{u_2}{v_2}(\alpha_2^{n-k-1}-\eta \alpha_2^{n-k}) & \ldots &\frac{ u_n}{v_n}(\alpha_n^{n-k-1}-\eta \alpha_n^{n-k}) \\
  \end{array}
\right)
\end{eqnarray}
is the check matrix of $C_k(\alpha, v, \eta)$,

  (3) Suppose that  $a\neq 0$ and $\eta = -a^{-1}$.  Then

\begin{eqnarray}\label{2.24}
  H_{n-k} &=& \tiny\left(
  \begin{array}{cccc}
   \frac{ u_1}{v_1} & \frac{u_2}{v_2} & \ldots & \frac{u_n}{v_n} \\
    \frac{u_1}{v_1}\alpha_1 & \frac{u_2}{v_2}\alpha_2 & \ldots & \frac{u_n}{v_n}\alpha_n \\
    \vdots & \vdots &  & \vdots \\
    \frac{u_1}{v_1}\alpha_1^{n-k-2} & \frac{u_2}{v_2}\alpha_2^{n-k-2} & \ldots & \frac{u_u}{v_n}\alpha_1^{n-k-2} \\
   \frac{ u_1}{v_1}\alpha_1^{n-k} & \frac{u_2}{v_2}\alpha_2^{n-k} & \ldots &\frac{ u_n}{v_n}\alpha_n^{n-k} \\
  \end{array}
\right)
\end{eqnarray}
is the check matrix of $C_k(\alpha, v, \eta)$,
\end{thm}
\begin{proof}
To calculate the check matrix of  $\mathcal C_{k}(\alpha, v, \eta)$, we  investigate  the check matrix of $\mathcal C_{k}(\alpha, 1, \eta)$.

There is a $n\times n$ matrix over $\Bbb F_q$:
 $$G=\left(
  \begin{array}{cccc}
    1 & 1& \ldots & 1\\
    \alpha_1 & \alpha_2 & \ldots & \alpha_n \\
     \vdots & \vdots &  & \vdots \\
    \alpha_1^{n-1} & \alpha_2^{n-1} & \ldots & \alpha_n^{n-1} \\
  \end{array}
\right).$$

We consider the system of equations over $\Bbb F_q$:   $G(u_1,u_2,\ldots,u_n)^T=(0,\ldots,0,1)^T$. Then there is an unique solution: $(u_1,\ldots, u_n)^T$, where  $u_i=\prod_{j=1,j\neq i}^n(\alpha_i-\alpha_j)^{-1}$, $1\le i\le n$.

    Let
$$H=\left(
  \begin{array}{cccccc}
   u_1\alpha_1^{n-1} & \cdots & u_1\alpha_1^{n-k-1}& \cdots & u_1\\
    u_2\alpha_2^{n-1}  & \cdots & u_2\alpha_1^{n-k-1}& \cdots & u_2\\
    \vdots  &  & \vdots &  & \vdots \\
    u_n\alpha_n^{n-1}  & \cdots & u_n\alpha_n^{n-k-1}& \cdots & u_n\\
  \end{array}
\right).$$
Then
$$GH=L=(l_{ij})_{1\leq i,j\leq n}$$
is a lower triangular matrix,  where    $l_{ii}=1$ for $1\le i\le n$, $ l_{ij}=0$ for $i< j$,   and
  $l_{k+1,k}=\sum_{i=1}^n\alpha_i^ku_i\alpha_i^{n-k}=\sum_{i=1}^n u_i\alpha_i^n=\sum_{i=1}^n\alpha_i$.
In fact, let  $l_{k+1,k}=a$ and
 $$m(x)=\prod_{j=1}^n(x-\alpha_j)=x^n+\sum_{j=0}^{n-1}a_jx^j.$$ Then by  $m(\alpha_i)=0$,  $\alpha_i^n=-\sum_{j=0}^{n-1}a_j\alpha_i^j$, $1 \leq i \leq n$.  Hence by $G(u_1,\ldots, u_{n-1}, u_n)^T=(0,\ldots, 0, 1)^T$,
 $$l_{k+1,k}=\sum_{i=1}^nu_i\alpha_i^n=-\sum_{j=0}^{n-1}a_j\sum_{i=0}^{n-1}u_i\alpha_i^j
 =-a_{n-1}=\sum_{i=1}^n\alpha_i=a.$$

Let $P_1=P(k, (k+1)(\eta))$ is an elementary matrix, whose $k$th row is  replaced by sum of  $\eta$ times $k+1$th row  and  $k$th row. Then there is an elementary matrix $P_2$  such that
\begin{equation}P_1GHP_2=L'=\left(
                               \begin{array}{cccccc}
                                 1 & \cdots & 0 & 0 & \cdots & 0 \\
                                 \vdots & \ddots & \vdots &  \vdots&  & \vdots \\
                                 * & \cdots & * & 0 & \cdots & 0 \\
                                 * & \cdots & * & * & \cdots & 0 \\
                                 \vdots &  & \vdots & \vdots & \ddots & \vdots \\
                                 * & \cdots&*& *& * & 1 \\
                               \end{array}
                             \right)\begin{array}{cccccc}\\ \\ k \\ k+1 \\ \\ \\ \end{array}\end{equation}
is a lower triangular matrix.

(1) Suppose that $a\ne 0$ and $1+a\eta\ne 0$. Then in (2.5),  we take  an elementary matrix $P_2=P(k, k+1(\frac {-\eta}{1+a\eta}))$,
whose  $k+1$th column is replaced by sum of $\frac{-\eta}{1+a\eta}$ times $k$th column and $k+1$th column.
Hence by (2.5),
$$
 H_{n-k}=\small  \left(
  \begin{array}{cccc}
    u_1 & u_2 & \ldots & u_n \\
    u_1\alpha_1 & u_2\alpha_2 & \ldots & u_n\alpha_n \\
    \vdots & \vdots &  & \vdots \\
    u_1\alpha_1^{n-k-2} & u_2\alpha_2^{n-k-2} & \ldots & u_1\alpha_1^{n-k-2} \\
    u_1(\alpha_1^{n-k-1}-\frac{\eta}{1+a\eta} \alpha_1^{n-k}) & u_2(\alpha_2^{n-k-1}-\frac{\eta}{1+a\eta} \alpha_2^{n-k}) & \ldots & u_n(\alpha_n^{n-k-1}-\frac{\eta}{1+a\eta} \alpha_n^{n-k}) \\
  \end{array}
\right)
$$
is a check matrix of  $C_k(\alpha, v, \eta)$.

(2)  Suppose that $a=0$ and $\eta\ne 0$. Then in (2.5), we take an elementary matrix $P_2=P(k, k+1(-\eta))$, whose $k+1$th column is replaced by sum of $-\eta$ times $k$th and $k+1$th column.
Similarly, we can prove the result.

(3) Suppose that $a\ne0$ and $1+a\eta=0$. Then in (2.5), we take an elementary matrix $P_2=P(k, k+1)$, that  exchanges $k$th column and $k+1$ column.
Similarly, we can prove the result.

\end{proof}

\subsection{MDS or NMDS TGRS codes}

 Recall that the definition of MDS codes and NMDS codes.
 \begin{defn}
 A $[n,k,d]$ linear code $\mathcal{C}$ over $\mathbb{F}_q$ is MDS if $d=n-k+1$.
 A $[n,k,d]$ linear code $\mathcal{C}$ over $\mathbb{F}_q$ is almost  MDS if $d=n-k$.
 A $[n,k,d]$ linear code $\mathcal{C}$ over $\mathbb{F}_q$ is NMDS if $\mathcal{C}$ and the dual of $\mathcal{C}$ are almost MDS codes, respectively.
 \end{defn}
Now, we present the sufficient and necessary condition that TGRS code is an MDS  code (see  \cite{1}). In addition, the sufficient and necessary condition of NMDS TGRS codes is also given in the following.

\begin{lem}\label{1} Let $\alpha_1,\ldots, \alpha_n$ are distinct elements in $\Bbb F_q$,  $\alpha=(\alpha_1,\ldots, \alpha_n)$,  $v=(v_1,\ldots, v_n)\in (\Bbb F_q^*)^n$, and  $\eta\in \Bbb F_q^*$. Let
\begin{equation}S_k=\{ \sum_{i \in I}\alpha_i:  \forall  I \subset \{1,\ldots,n\}, |I|=k\}.\end{equation}
Then

$(1)$ the  TGRS code $\mathcal C_{k}(\alpha, v, \eta)$ is MDS if and only if $-\eta^{-1} \notin S_k$;

$(2)$ the  TGRS code $\mathcal C_{k}(\alpha, v, \eta)$ is NMDS if and only if $-\eta^{-1} \in S_k$.
\end{lem}
\begin{proof}
$(1)$ $C_{k}(\alpha, v, \eta)$ is MDS
$\Longleftrightarrow$ any $k$ columns of generator matrix of $C_{k}(\alpha, v, \eta)$ are linear independently

 \begin{eqnarray*}\Longleftrightarrow
&&\left |
  \begin{array}{cccc}
    1 &1 & \ldots &1 \\
    \alpha_{i_1} & \alpha_{i_2} & \ldots & \alpha_{i_k} \\
    \vdots & \vdots &  & \vdots \\
    \alpha_{i_1}^{k-1}+\eta \alpha_{i_1}^k & \alpha_{i_2}^{k-1}+\eta \alpha_{i_2}^k & \ldots & \alpha_{i_k}^{k-1}+\eta \alpha_{i_k}^k \\
  \end{array}
\right |\\ && =\prod_{1\le s<t\le k}(\alpha_{i_t}-\alpha_{i_s})(1+\eta\sum_{s=1}^t\alpha_{i_s})\neq 0,
\end{eqnarray*}
where $\{i_1,i_2,\ldots,i_k\}$ is an arbitrary $k$-subset of $\{1,2,\ldots,n\}$. Then the result follows immediately.

$(2)$ $``\Longleftarrow"$ By linear algebra, we know that any $k-1$ columns of $G_k$ are linear independently over $\mathbb{F}_q$.  If $-\eta^{-1} \in S_k$, then there exists $k$ columns of $G_k$ are linear dependently over $\mathbb{F}_q.$ Thus, the parameter of $C_{k}(\alpha, v, \eta)^\perp$ is $[n,n-k,k]$. Similarly,  since any $n-k-1$ columns of $H_k$ are linear independently over $\mathbb{F}_q$ and $C_{k}(\alpha, v, \eta)^\perp$ is not MDS, we obtain the parameter of $C_{k}(\alpha, v, \eta)$ is $[n,k,n-k]$. Thus, $C_{k}(\alpha, v, \eta)$ is NMDS.

$``\Longrightarrow"$ Conversely, if $C_{k}(\alpha, v, \eta)$ is NMDS, then the parameter of $C_{k}(\alpha, v, \eta)^\perp$ is $[n,n-k,k]$, which implies that there exists $k$ columns of $G_k$ is linear dependently over $\mathbb{F}_q$, i.e. $-\eta^{-1} \in S_k$.
\end{proof}

\subsection{Self-dual TGRS codes}
For any vectors $a=(a_1,\ldots,a_n), b=(b_1,\ldots, b_n)\in \mathbb{F}_q^n$, the inner product is defined as $\langle a,b\rangle=\sum_{i=1}^{n}a_ib_i$. Then for the $q$-ary $[n,k,d]$ code $C$, the dual code $C^{\perp}=\{c \in \mathbb{F}_q^n : \langle c,a\rangle =0, \forall a \in C\}$. A code with $C= C^\perp$ is called self-dual. Next, we will investigate self-dual TGRS codes.

Let $G_k$ be  the generator matrix  of  $\mathcal C_{k}(\alpha, v, \eta)$ as (\ref{2.1})
and, by different value of $a,\eta$, $H_{n-k}$ be  the check matrix  of  $\mathcal C_{k}(\alpha, v, \eta)$ as (\ref{2.2}), (\ref{2.22}) and (\ref{2.24}), respectively.  Then we have the following lemma, which plays an important role in constructing self-dual codes.
\begin{thm}\label{22222}
Let $n=2k$ with $k (\geq 3)$ a positive integer and $\eta \in \mathbb{F}_q^*.$ Let $G_k$ be  the generator matrix  of  $\mathcal C_{k}(\alpha, v, \eta)$, where
\begin{eqnarray*}
  G_k &=&\left(
  \begin{array}{cccc}
    v_1 & v_2 & \ldots & v_n \\
    v_1\alpha_1 & v_2\alpha_2 & \ldots & v_n\alpha_n \\
     \vdots & \vdots &  & \vdots \\
    v_1(\alpha_1^{k-1}+\eta \alpha_1^k) & v_2(\alpha_2^{k-1}+\eta \alpha_2^k) & \ldots & v_n(\alpha_n^{k-1}+\eta \alpha_n^k) \\
  \end{array}
\right)
\end{eqnarray*}
with $v_i\neq 0$ for $1 \leq i \leq n$ and  $\alpha_i, 1 \leq i \leq n$ are $n$ distinct elements in $\mathbb{F}_q$.

\begin{enumerate}
  \item [(a)] Suppose that $a\ne 0$ and $\eta \neq -a^{-1}$.
  Then $\mathcal C_{k}(\alpha, v, \eta)$ is self-dual if and only if the following two conditions hold: \\
(1) there exists $\lambda \in \mathbb{F}_q^*$ such that $v_i^2=\lambda u_i$ for all $1 \leq i \leq n$, where $u_i=\prod_{j \neq i, j=1}^n(\alpha_i-\alpha_j)^{-1}$, and\\
(2) $\eta=\frac{-\eta}{1+a\eta}$, i.e. $2+a\eta=0$.

  \item [(b)] Suppose that  $a=0$ and $\eta \neq 0$.  Then $\mathcal C_{k}(\alpha, v, \eta)$ can not be a self-dual code.

  \item [(c)] Suppose that $a\neq 0$ and $\eta = -a^{-1}$.  Then $\mathcal C_{k}(\alpha, v, \eta)$ can not be a self-dual code.

\end{enumerate}
\end{thm}
\begin{proof}
Let $H_{n-k}$ be the check matrix of $\mathcal C_{k}(\alpha, v, \eta)$. For simplification, write
$$G_k=\tiny\left(
       \begin{array}{c}
         \beta_0 \\
         \beta_1 \\
         \vdots \\
         \beta_{k-1} \\
       \end{array}
     \right)
,H_{n-k}=\tiny\left(
          \begin{array}{c}
            \gamma_0 \\
            \gamma_1 \\
            \vdots \\
            \gamma_{n-k-1} \\
          \end{array}
        \right)
.$$
If $\mathcal C_{k}(\alpha, v, \eta)$ is self-dual, then $Span_{\mathbb{F}_q}\{\beta_0,\beta_1,\ldots,\beta_{k-1}\}=Span_{\mathbb{F}_q}\{\gamma_0,\gamma_1,\ldots,\gamma_{k-1}\}$.

 $(a)$ Let $a\ne 0$ and $\eta \neq -a^{-1}$, then
 \begin{eqnarray*}
  H_{n-k} &=& \tiny\left(
  \begin{array}{cccc}
   \frac{ u_1}{v_1}  & \ldots & \frac{u_n}{v_n} \\
    \frac{u_1}{v_1}\alpha_1  & \ldots & \frac{u_n}{v_n}\alpha_n \\
    \vdots  &  & \vdots \\
    \frac{u_1}{v_1}\alpha_1^{n-k-2}  & \ldots & \frac{u_u}{v_n}\alpha_1^{n-k-2} \\
   \frac{ u_1}{v_1}(\alpha_1^{n-k-1}- \frac{\eta }{1+a\eta}\alpha_1^{n-k}) & \ldots &\frac{ u_n}{v_n}(\alpha_n^{n-k-1}-\frac{\eta }{1+a\eta}\alpha_n^{n-k}) \\
  \end{array}
\right).
\end{eqnarray*}

 $``\Longrightarrow"$:
 If $\mathcal C_{k}(\alpha, v, \eta)$ is self-dual, then
for any $\beta_i, 0 \leq i \leq k-1$, we have $\beta_i \in Span_{\mathbb{F}_q}\{\gamma_0,\gamma_1,\ldots,\gamma_{k-1}\}$. Particularly,
 $\beta_0=(a_0,a_1,\ldots,a_{k-1})H_{n-k}$ with $a_0,a_1,\ldots,a_{k-1} $ not all zero elements in $\mathbb{F}_q$, i.e.
there exists $f(x)=a_0+a_1x+\ldots+a_{k-1}(x^{k-1}-\frac{\eta}{1+a \eta} x^k)\in \mathbb{F}_q[x]$  such that $$\frac{v_i^2}{u_i}=f(\alpha_i),~1 \leq i \leq n.$$  Moreover,
$\beta_{k-2}=(b_0,b_1,\ldots,b_{k-1})H_{n-k}$ with $b_0,b_1,\ldots,b_{k-1} $ not all zero elements in $\mathbb{F}_q$, i.e.
there exists $g(x)=b_0+b_1x+\ldots+b_{k-1}(x^{k-1}-\frac{\eta}{1+a \eta}x^k)$  such that $$f(\alpha_i)\alpha_i^{k-2}=\frac{v_i^2}{u_i}\alpha_i^{k-2}=g(\alpha_i),~1\leq i \leq n.$$
Noting that $deg(f(x)x^{k-2}-g(x))\leq n-2$ and $\alpha_i( 1\leq i \leq n)$ are different roots of  $f(x)x^{k-2}-g(x)$, we then obtain  $f(x)x^{k-2}-g(x)=0$.
Consequently, coefficients of  $f(x)x^{k-2}-g(x)$ are equal to $0$. We then obtain

$$\left\{
  \begin{array}{ll}
   a_0=b_{k-2} , & \hbox{} \\
    a_1=b_{k-1}, & \hbox{} \\
    a_2=-\frac{\eta}{1+a\eta} b_{k-1}, & \hbox{} \\
    a_i=0, & \hbox{$3 \leq i \leq n$.}
  \end{array}
\right.$$
Thus, $f(x)=a_0+a_1x-\frac{\eta}{1+a\eta}a_1x^2$.
Noting that
$\beta_{k-1}=(c_0,c_1,\ldots,c_{k-1})H_{n-k}$ with $c_0,c_1,\ldots,c_{k-1} $ not all zero elements in $\mathbb{F}_q$, i.e.
 there exists $h(x)=c_0+c_1x+\ldots+c_{k-1}(x^{k-1}-\frac{\eta}{1+a \eta}  x^k)\in \mathbb{F}_q[x]$  such that $$f(\alpha_i)(\alpha_i^{k-1}+\eta \alpha_i^k)=\frac{v_i^2}{u_i}(\alpha_i^{k-1}+\eta \alpha_i^k)=h(\alpha_i),~1 \leq i \leq n.$$
Since $k \geq 3$, we have $deg(f(x)(x^{k-1}+\eta x^k)-h(x))< n$. Consequently, $f(x)(x^{k-1}+\eta x^k)-h(x)=0$.
We then obtain
$$\left\{
  \begin{array}{ll}
   a_0=c_{k-1},  & \hbox{} \\
    -\frac{\eta}{1+a\eta}a_1\eta=0.  & \hbox{}
  \end{array}
\right.$$
Since $\eta \neq 0,$ we have $a_1=0$. That is, $f(x)=a_0 \neq 0$.
Thus, $\frac{v_i^2}{u_i}=f(\alpha_i)=a_0$ for $1 \leq i \leq n$.

Suppose $\eta \neq \frac{-\eta}{1+a \eta}$. Since
$$\beta_{k-1}-\lambda \gamma_{k-1}=(\eta+\frac{\eta}{1+a \eta})(v_1 \alpha_1^k,v_2 \alpha_2^k,\ldots,v_n \alpha_n^k) \in \mathcal{C}_k(\alpha,v,\eta),$$ then $(v_1\alpha_1^k,v_2\alpha_2^k,\ldots,v_n\alpha_n^k) \in \mathcal{C}_k(\alpha,v,\eta)$. Consequently, $(v_1\alpha_1^{k-1},v_2\alpha_2^{k-1},\ldots,v_n\alpha_n^{k-1}) \in \mathcal{C}_k(\alpha,v,\eta)$, which is a contradiction.

$``\Longleftarrow"$: It is obvious, so we omit it here.

$(b)$ Let $a= 0$ and $\eta \neq 0$, then $H_{n-k}$ is the form as (\ref{2.22}).
If $C_k(\alpha,v,\eta)$ is self-dual, then $\beta_{k-1}+\lambda \gamma_{k-1}\in \mathcal{C}_k(\alpha,v,\eta)$, i.e., $(v_1\alpha_1^{k-1},v_2\alpha_2^{k-1},\ldots,v_n\alpha_n^{k-1}) \in \mathcal{C}_k(\alpha,v,\eta)$. Consequently, $(v_1\alpha_1^{k},v_2\alpha_2^{k},\ldots,v_n\alpha_n^{k}) \in \mathcal{C}_k(\alpha,v,\eta)$, which is a contradiction.

$(c)$ Let $a\neq 0$ and $\eta = -a^{-1}$,  then $H_{n-k}$ is the form as (\ref{2.24}).
If $C_k(\alpha,v,\eta)$ is self-dual, then $\lambda \gamma_{k-1}\in \mathcal{C}_k(\alpha,v,\eta)$, i.e., $(v_1\alpha_1^{k},v_2\alpha_2^{k},\ldots,v_n\alpha_n^{k}) \in \mathcal{C}_k(\alpha,v,\eta)$. Consequently, $(v_1\alpha_1^{k-1},v_2\alpha_2^{k-1},\ldots,v_n\alpha_n^{k-1}) \in \mathcal{C}_k(\alpha,v,\eta)$, which is a contradiction.

This completes the proof.
\end{proof}

\section{Self-dual MDS or self-dual NMDS codes}

In this section, we mainly construct several classes of self-dual codes. Thus, we only consider the TGRS codes in Theorem \ref{22222}$(a)$,   which is a either MDS code or NMDS code. Note that $u_i=\prod_{j=1,j \neq i}^{n}(\alpha_i-\alpha_j)^{-1}$ for $1 \leq i \leq n$.

\subsection{Self-dual codes over $\mathbb{F}_q$} We will construct self-dual MDS or self-dual NMDS codes from TGRS codes over $\mathbb{F}_q$, where $q$ is an odd prime power.
\begin{thm}\label{2}
Let $s,t,l$ be positive integers, $p$ an odd prime and $q=p^s, q_1=p^t$.
  Assume that $\mathbb{F}_q$ is the splitting field of $m(x)$ over $\mathbb{F}_{q_1}$, where $m(x)=x^{2lp}+bx^{2lp-1}+c$, $b,c \in \mathbb{F}_{q_1}^*$.
Let $\alpha_i$ be the root of $m(x)$ and  $v_i=\alpha_i^{1-lp}$ for $1 \leq i \leq 2lp$. Write $\alpha=(\alpha_1,\ldots,\alpha_{2lp})$, $v=(v_1,\ldots,v_{2lp})$ and $\eta=2b^{-1}$, then
$$C_{lp}(\alpha, v, \eta)=\{(v_1f(\alpha_1),\ldots,v_{2lp}f(\alpha_{2lp}))\mid f(x)=\sum_{i=0}^{lp-1}f_ix^i+\eta f_{lp-1}x^{lp}\in \mathbb{F}_{q}[x] \}$$
is a self-dual TGRS code of length $2lp$ over $\mathbb{F}_{q}$.
\end{thm}

\begin{proof}
Since $m'(x)=-bx^{2lp-2}$, we then have $gcd(m(x),m'(x))=1$. Thus, $\alpha_1, \ldots, \alpha_{2lp}$ are distinct elements in $\mathbb{F}_q$. Since  $$u_i=m'(\alpha_i)^{-1}=-b^{-1}\alpha_i^{2-2lp}=-b^{-1}v_i^2\neq 0$$ for $1 \leq i \leq 2lp$.
Moreover, $2-b\eta=0$.
By Theorem \ref{22222}, we have
$$C_{lp}(\alpha, v, \eta)=\{(v_1f(\alpha_1),\ldots,v_{2lp}f(\alpha_{2lp}))\mid f(x)=\sum_{i=0}^{lp-1}f_ix^i+\eta f_{lp-1}x^{lp}\in \mathbb{F}_{q}[x] \}$$
is a self-dual TGRS code of length $2lp$ over $\mathbb{F}_{q}$.
\end{proof}

\begin{cor}
In Theorem \ref{2}, if $\sum_{i\in I}\alpha_i \neq -\frac{b}{2}$ for any
$I \subset \{1,2,\ldots,2lp \}$ with $|I|=lp$, then $C_{lp}(\alpha, v, \eta)$ is a self-dual MDS  TGRS code of length $2lp$ over $\mathbb{F}_{q}$ by Lemma \ref{1}. Otherwise,  $C_{lp}(\alpha, v, \eta)$ is a  self-dual NMDS TGRS code of length $2lp$ over $\mathbb{F}_{q}$.
\end{cor}

\begin{thm}\label{3}
Assume that $q=p^s$ with $p(> 3)$ an odd prime, $s$ a positive integer and $\mathbb{F}_q$ the splitting field of $m(x)$ over $\mathbb{F}_p$, where
 $m(x)=x^{2p}-x^{2p-1}+2x^{p+1}+3^{-1}x^3+1$. Let $gcd(m'(x),x^4-(3-3^{-1})x^3+1)=1$, $m(\alpha_i)=0$ and  $v_i=(\alpha_i^{p-1}+\alpha_i)^{-1}$ for $1 \leq i \leq 2p$. Write $\alpha=(\alpha_1,\ldots,\alpha_{2p})$, $v=(v_1,\ldots,v_{2p})$, then
$$C_{p}(\alpha, v, \eta)=\{(v_1f(\alpha_1),\ldots,v_{2p}f(\alpha_{2p}))\mid f(x)=\sum_{i=0}^{p-1}f_ix^i-2 f_{p-1}x^p\in \mathbb{F}_{q}[x] \}$$
is a self-dual TGRS code of length $2p$ over $\mathbb{F}_{q}$.
\end{thm}
\begin{proof}
$m'(x)=x^{2p-2}+2x^p+x^2=(x^{p-1}+x)^2$. Suppose $(m(x),m'(x))\neq 1$, then there exists $\beta \neq 0$ such that  $m'(\beta)=0=m(\beta)$. That implies $\beta^{p-2}=-1$, then $m(\beta)=\beta^{2p}-\beta^{2p-1}+2\beta^{p+1}+3^{-1}\beta^3+1=
\beta^4-\beta^3-2\beta^3+3^{-1}\beta^3+1=0$, which is a contradiction with
$gcd(m'(x),x^4-(3-3^{-1})x^3+1)=1$. Thus, $(m(x),m'(x))= 1$, which implies $\alpha_1,\ldots, \alpha_{2p}$ are distinct elements in $\mathbb{F}_q$. Moreover, $$u_i=m'(\alpha_i)^{-1}=(\alpha_i^{p-1}+\alpha_i)^{-2}=v_i^2\neq 0$$ for $1 \leq i \leq 2p$. By Theorem \ref{22222},  $C_{p}(\alpha, v, \eta)$ is self-dual.
\end{proof}

\begin{cor}
In Theorem \ref{3}, if $\sum_{i\in I}\alpha_i \neq \frac{1}{2}$ for any
$I \subset \{1,2,\ldots,2p \}$ with $|I|=p$, then $C_{p}(\alpha, v, \eta)$ is a self-dual MDS  TGRS code of length $2p$ over $\mathbb{F}_{q}$ by Lemma \ref{1}. Otherwise,  $C_{p}(\alpha, v, \eta)$ is a  self-dual NMDS TGRS code of length $2p$ over $\mathbb{F}_{q}$.
\end{cor}

\subsection{Self-dual codes over $\mathbb{F}_{q^2}$}

We will construct self-dual TGRS codes over $\mathbb{F}_{q^2}$, where $q$ is an odd prime power.
We need the following lemma first, which is very basic but important.

\begin{lem} Each element in $\mathbb{F}_q$ is a square element in $\mathbb{F}_{q^2}$, where $q$ is odd.
\end{lem}
\begin{proof}
Let $\mathbb{F}_{q^2}^*=\langle \alpha\rangle$. Then  $\mathbb{F}_{q}^*=\langle \alpha^{q+1}\rangle$, $2|q+1$. It is right.
\end{proof}

Next we turn to construct self-dual MDS or self-dual NMDS codes from TGRS codes. We obtain the following theorem.

\begin{thm}\label{5} Assume that  $n $ is an even integer with $\gcd(q, n)=1$,
  $m(x)=x^{n}+bx^{n-1}+1\in \Bbb F_{q'}[x]$, $b\ne 0$,    and $\mathbb{F}_q$  the splitting field of $m(x)$ over $\Bbb F_{q'}$.
    Let $\alpha_i$ be the roots of $m(x)$ which satisfying $\alpha_i\neq b(1-n)n^{-1}$ and $\alpha_i^n \neq n-1$, then $\alpha_1,\ldots, \alpha_n$ are pairwise distinct and $v_i\neq 0$ for $1 \leq i \leq n$.
     Write $\alpha=(\alpha_1,\ldots,\alpha_n)$, $v=(v_1,\ldots,v_n)$ and $\eta=2 b^{-1}$, where $$\{v_i \in \mathbb{F}_{q^2} \mid v_i^2=\frac{\alpha_i}{\alpha_i^n-n+1}, 1 \leq i \leq n\}.$$
    Then
$$C_{\frac{n}{2}}(\alpha, v, \eta)=\{(v_1f(\alpha_1),\ldots,v_nf(\alpha_n))\mid f(x)=\sum_{i=0}^{\frac{n}{2}-1}f_ix^i+\eta f_{\frac{n}{2}-1}x^{\frac{n}{2}}\in \mathbb{F}_{q^2}[x] \}$$
is a self-dual TGRS code of length $n$ over $\mathbb{F}_{q^2}$.

\end{thm}
\begin{proof}
There is an irreducible factorization of $m(x)$ over $\mathbb{F}_q$, i.e.  $$m(x)=x^{n}+bx^{n-1}+1=\prod_{i=1}^n(x-\alpha_i)\in \mathbb{F}_q[x].$$
It is easy to obtain that $\frac{b(1-n)}{n}$ and $0$ are roots of $m'(x)=nx^{n-1}+b(n-1)x^{n-2}$. Since $\alpha_i\neq b(1-n)n^{-1}$ and $m(0)\neq 0$
we then have $gcd(m(x),m'(x))=1$. Consequently, $m(x)=0$ has no repeated roots, i.e. $\alpha_i(1 \leq i \leq n)$ are distinct. Since $\alpha_i^n \neq n-1$ then $v_i\neq 0$.
Thus, $C_{\frac{n}{2}}(\alpha, v, \eta)$ is a TGRS code over $\mathbb{F}_{q^2}$. By $$u_i=m'(\alpha_i)^{-1}=\frac{\alpha_i}{n\alpha_i^n+(n-1)b \alpha_i^{n-1}}=\frac{\alpha_i}{n\alpha_i^n-(n-1)(\alpha_i^n+1)}=\frac{\alpha_i}{\alpha_i^n-n+1}
=v_i^2$$ for $1 \leq i \leq n$ and $2-b\eta=0$, according to Theorem \ref{22222}, $C_{\frac{n}{2}}(\alpha, v, \eta)$ is self-dual.
\end{proof}

\begin{cor}
In Theorem \ref{5}, if $\sum_{i\in I}\alpha_i \neq -\frac{b}{2}$ for any
$I \subset \{1,2,\ldots,n \}$ with $|I| =\frac{n}{2}$, then $C_{\frac{n}{2}}(\alpha, v, \eta)$ is a self-dual MDS TGRS code of length $n$ over $\mathbb{F}_{q^2}$ by Lemma \ref{1}. Otherwise,  $C_{\frac{n}{2}}(\alpha, v, \eta)$ is a  self-dual NMDS TGRS code of length $n$ over $\mathbb{F}_{q^2}$.
\end{cor}

\begin{thm} \label{4} Let $\beta\in \mathbb{F}_q^*$, $\lambda=\beta^{n+1}$ and $n+1 | q-1$ with $n$ an even integer.
 Assume that $m(x)=x^n+x^{n-1}\beta+\ldots+x\beta^{n-1}+\beta^n$ and $\alpha_i, 1 \leq i \leq n$ the roots of $m(x)$.
Let $v_i\in \mathbb{F}_{q^2}$ such that $v_i^2=\alpha_i(\alpha_i-\beta)$ for $1 \leq i \leq n$ and $\eta=2 \beta^{-1}$. Write $\alpha=(\alpha_1,\ldots,\alpha_n)$, $v=(v_1,\ldots,v_n)$, then $$C_{\frac{n}{2}}(\alpha, v, \eta)=\{(v_1f(\alpha_1),\ldots,v_nf(\alpha_n))\mid f(x)=\sum_{i=0}^{\frac{n}{2}-1}f_ix^i+\eta f_{\frac{n}{2}-1}x^{\frac{n}{2}}\in \mathbb{F}_{q^2}[x] \}$$
is a self-dual TGRS code of length $n$ over $\mathbb{F}_{q^2}$.
\end{thm}
\begin{proof}
Let $ \gamma \in \mathbb{F}_q^*$ such that $ord(\gamma)=n+1$, then there is a irreducible factorization of $m(x)$ over $\mathbb{F}_q$, i.e.  $$m(x)=\frac{x^{n+1}-\lambda}{x-\beta}=
\prod_{i=0}^n(x-\beta\gamma^i).$$ Without loss of generality, write $\alpha_i=\beta \gamma^i$, then  $m'(\alpha_i)=\frac{(n+1)\lambda }{\alpha_i(\alpha_i-\beta)}$ for $1\leq i \leq n$.
Then we have $$u_i=m'(\alpha_i)^{-1}=(n+1)^{-1}\lambda^{-1} v_i^2.$$
By Lemma \ref{3},
 then
$$C_{\frac{n}{2}}(\alpha, v, \eta)=\{(v_1f(\alpha_1),\ldots,v_nf(\alpha_n))\mid f(x)=\sum_{i=1}^{\frac{n}{2}-1}f_ix^i+\eta f_{\frac{n}{2}-1}x^{\frac{n}{2}}\in \mathbb{F}_{q^2}[x] \}$$
is a self-dual TGRS code of length $n$ over $\mathbb{F}_{q^2}$.
\end{proof}

\begin{rem}
In Theorem \ref{4}, let $m(x)=x^n+x^{n-1}(\beta\gamma^j)+\ldots+x(\beta\gamma^j)^{n-1}+(\beta\gamma^j)^n
=\frac{x^{n+1}-\lambda}{x-\beta\gamma^j}$ with $ord(\gamma)=n+1, 0 \leq j \leq n.$ And $\alpha_i,~1 \leq i \leq n$ are the roots of $m(x)$.
Let $v_i\in \mathbb{F}_{q^2}$ such that  $v_i^2=\alpha_i(\alpha_i-\beta \gamma^j)$ for $1 \leq i \leq n$ and $\eta=2 (\beta\gamma^j)^{-1}$. Then self-dual TGRS codes of length $n$ over $\mathbb{F}_{q^2}$ could be constructed as well.
\end{rem}

\begin{cor}
In Theorem \ref{4}, if $\sum_{i\in I}\alpha_i \neq -\frac{\beta}{2}$ for any
$I \subset \{1,2,\ldots,n \}$ with $  |I| =\frac{n}{2}$, then $C_{\frac{n}{2}}(\alpha, v, \eta)$ is a self-dual MDS TGRS code of length $n$ over $\mathbb{F}_{q^2}$. Otherwise,  $C_{\frac{n}{2}}(\alpha, v, \eta)$ is a self-dual NMDS TGRS code of length $n$ over $\mathbb{F}_{q^2}$.
\end{cor}

\subsection{Some examples}

Next we will give examples of self-dual MDS TGRS codes and self-dual NMDS TGRS codes according to  Theorem
\ref{4}.

\begin{exa}
 $\mathbb{F}_{89}^*=\langle 3\rangle$.

 $(1)$
let $\alpha_i \equiv (3^8)^i\equiv 64^i \pmod {89}$ for $1 \leq i \leq 10$, we then have
$$\frac{x^{11}-1}{x-1}=\prod_{i=1}^{10}(x-\alpha_i).$$
  Let $u_i=\prod_{j=1,j \neq i}^{4}(\alpha_i-\alpha_j)^{-1}$, then there exists $v_i\in \mathbb{F}_{61^2}$ such that $u_i=v_i^2$ for $1 \leq i \leq 10$.

  $(2)$ Write $a=\sum_{i=1}^{10}\alpha_i$. Let $\eta=2$, then $2+a \eta=0$.

  $(3)$ Note that $\sum_{i\in I_5}\alpha_i \neq  \frac{-1}{\eta}=44$  for any $I_5 \subset \{1,2,\ldots,10\}$ with $\mid I_5\mid=5$.
  Thus,
let $\alpha=(\alpha_1,\ldots,\alpha_{10})$, $v=(v_1,\ldots,v_{10})$, then
$$C_{5}(\alpha, v, 2)=\{(v_1f(\alpha_1),\ldots,v_{10}f(\alpha_{10}))\mid f(x)=\sum_{i=0}^{4}f_ix^i+2 f_{4}x^5\in \mathbb{F}_{89^2}[x] \}$$
is a  self-dual MDS TGRS code of length $10$ over $\mathbb{F}_{89^2}$.
\end{exa}

\begin{exa}
 $\mathbb{F}_{61}^*=\langle 2\rangle$.

 $(1)$
let $\alpha_i \equiv (2^4)^i \equiv 16^i \pmod {61}$ for $1 \leq i \leq 14$,
we then have $$\frac{x^{15}-1}{x-1}=\prod_{i=1}^{14}(x-\alpha_i).$$    Let $u_i=\prod_{j=1,j \neq i}^{4}(\alpha_i-\alpha_j)^{-1}$, then there exists $v_i\in \mathbb{F}_{61^2}$ such that $u_i=v_i^2$ for $1 \leq i \leq 14$.

  $(2)$ Write $a=\sum_{i=1}^{14}\alpha_i$, let $\eta=2$, then $2+a \eta=0$.

  $(3)$ Note that $\sum_{i\in I_7}\alpha_i= \frac{-1}{\eta}=30$ if  $I_7 =\{ 2,5,9,11,12,13,14 \}$ with $\mid I_7\mid=7$.
  Thus,
let $\alpha=(\alpha_1,\ldots,\alpha_{14})$, $v=(v_1,\ldots,v_{14})$, then
$$C_{7}(\alpha, v, 2)=\{(v_1f(\alpha_1),\ldots,v_{14}f(\alpha_{14}))\mid f(x)=\sum_{i=0}^{6}f_ix^i+2 f_{6}x^7\in \mathbb{F}_{61^2}[x] \}$$
is a  self-dual  NMDS TGRS code of length $14$ over $\mathbb{F}_{61^2}$.
\end{exa}

\section{Conclusion}
In this paper, we investigate self-dual MDS and self-dual NMDS codes by TGRS codes. We give the check matrices of TGRS codes, which play an important role in investigating dual codes of TGRS codes. And we give the efficient and necessary condition of self-dual TGRS codes.  By factorization of several polynomials over finite field, we decide  $\alpha$, $v$ and $\eta$ such that $C_{\frac{n}{2}}(\alpha,v,\eta)$ are self-dual. Consequently, we obtain several classes of self-dual MDS codes. It is possible to construct more classes self-dual MDS or self-dual NMDS codes by different polynomials from TGRS codes.

\end{document}